\documentclass[a4paper]{article}

\usepackage{amsmath}
\usepackage{amssymb}
\usepackage{amsthm}
\usepackage{mathrsfs}
\usepackage{amscd}
\usepackage{algorithm}
\usepackage{algorithmicx}
\usepackage{upgreek}
\usepackage[hyphens]{url}
\usepackage{graphicx}
\usepackage{hyperref}
\usepackage{url}
\usepackage{breakurl}
\usepackage{float}
\usepackage[noend]{algpseudocode}
\usepackage{verbatim}
\usepackage{pbox}

\usepackage[margin=1.2in]{geometry}

\newtheorem{theorem}{Theorem}

\usepackage{color}

\definecolor{dblackcolor}{rgb}{0.0,0.0,0.0}
\definecolor{dbluecolor}{rgb}{0.01,0.02,0.7}
\definecolor{dgreencolor}{rgb}{0.2,0.4,0.0}
\definecolor{dgraycolor}{rgb}{0.30,0.3,0.30}

\algrenewcommand\algorithmicrequire{\textbf{Precondition:}}
\algrenewcommand\algorithmicensure{\textbf{Postcondition:}}

\usepackage{listings}
\lstdefinelanguage{Sage}[]{Python}
{morekeywords={False,sage,True},sensitive=true}
\renewcommand{\algorithmicrequire}{\textbf{Input:}}
\renewcommand{\algorithmicensure}{\textbf{Output:}}

\begin{document}

\title{Rigorous high-precision computation of the Hurwitz zeta function and its derivatives}

\author{%
  Fredrik Johansson
  \footnote{Supported by the Austrian Science Fund (FWF) grant Y464-N18.}\\[\medskipamount]
  \small{\strut RISC} \\
  \small{\strut Johannes Kepler University} \\
  \small{\strut 4040 Linz, Austria} \\
  \small{\strut fredrik.johansson@risc.jku.at}
}

\date{}

\maketitle

\begin{abstract}
We study the use of the Euler-Maclaurin formula to numerically evaluate the
Hurwitz zeta function $\zeta(s,a)$ for $s, a \in \mathbb{C}$, along with
an arbitrary number of derivatives with respect to $s$, to arbitrary
precision with rigorous error bounds. Techniques that lead to a fast
implementation are discussed. We present new record computations of
Stieltjes constants, Keiper-Li coefficients and the first nontrivial
zero of the Riemann zeta function, obtained using an open source
implementation of the algorithms described in this paper.
\end{abstract}

\section{Introduction}

The Hurwitz zeta function $\zeta(s,a)$ is defined for complex
numbers $s$ and $a$ by analytic continuation of the sum
\begin{equation*}
\zeta(s,a) = \sum_{k=0}^{\infty} \frac{1}{(k+a)^s}.
\end{equation*}
The usual Riemann zeta function is given by $\zeta(s) = \zeta(s,1)$.

In this work, we consider numerical computation of $\zeta(s,a)$ by the
Euler-Maclaurin formula with rigorous error control. Error bounds for
$\zeta(s)$ are classical (see for example \cite{Edwards1974},
\cite{BorweinBradleyCrandall2000} and numerous references therein),
but previous works have restricted to the case $a = 1$ or have not
considered derivatives. Our main contribution is to give an efficiently
computable error bound for $\zeta(s,a)$ valid for any complex $s$ and $a$
and for an arbitrary number of derivatives with respect to $s$
(equivalently, we allow $s$ to be a formal power series).

We also discuss implementation aspects, such as parallelization and
use of fast polynomial arithmetic. An open source implementation of
$\zeta(s,a)$ based on the algorithms described in this paper is
available. In the last part of the paper, we present results from
some new record computations done with this implementation.

Our interest is in evaluating $\zeta(s,a)$ to high precision
(hundreds or thousands of digits) for a single~$s$ of moderate height,
say with imaginary part less than $10^6$. Investigations of zeros of
large height typically use methods based on the Riemann-Siegel formula
and fast multi-evaluation techniques such as the
Odlyzko-Sch\"{o}nhage algorithm \cite{Odlyzko1988fast} or
the recent algorithm of Hiary \cite{Hiary2011fast}.

This work is motivated by several applications. For example, recent work
of Matiyasevich and Beliakov required values of thousands of nontrivial
zeros $\rho_n$ of $\zeta(s)$ to a precision of several thousand digits
\cite{Matiyasevich2012, MatiyasevichBeliakov2011}. Investigations of
quantities such as the Stieltjes constants $\gamma_n(a)$ and the Keiper-Li
coefficients $\lambda_n$ also call for high-precision
values \cite{Keiper1992power, Kreminski2003}. The difficulty is not
necessarily that the final result needs to be known to very high
accuracy, but that intermediate calculations may involve catastrophic
cancellation.

More broadly, the Riemann and Hurwitz zeta functions are useful for
numerical evaluation of various other special functions such as polygamma
functions, polylogarithms, Dirichlet $L$-functions, generalized hypergeometric
functions at singularities \cite{Bogolubsky2006}, and certain
number-theoretical constants~\cite{Flajolet1996zeta}. High-precision numerical
values are of particular interest for guessing algebraic relations among
special values of such functions (which subsequently may be proved rigorously
by other means) or ruling out the existence of algebraic relations with
small norm~\cite{Bailey2000experimental}.

\section{Evaluation using the Euler-Maclaurin formula}

Assume that $f$ is analytic on a domain containing $[N, U]$ where
$N, U \in \mathbb{Z}$, and let $M$ be a positive integer. Let $B_n$ denote
the $n$-th Bernoulli number and let $\tilde B_n(t) = B_n(t-\lfloor t\rfloor)$
denote the $n$-th periodic Bernoulli polynomial. The Euler-Maclaurin
summation formula (described in numerous works, such as \cite{Olver1997})
states that
\begin{equation}
\sum_{k=N}^{U} f(k) = I + T + R
\label{eq:emformula}
\end{equation}
where
\begin{align}
I & = \int_N^{U} f(t) \,dt, \\
T & = \frac{1}{2} \left(f(N) + f(U)\right) + \sum_{k=1}^{M} \frac{B_{2k}}{(2k)!} \left( f^{(2k-1)}(U) - f^{(2k-1)}(N) \right), \\
R & = -\int_N^{U} \frac{\tilde B_{2M}(t)}{(2M)!} f^{(2M)}(t) \,dt.
\label{eq:emformulaR}
\end{align}
If $f$ decreases sufficiently rapidly, \eqref{eq:emformula}--\eqref{eq:emformulaR}
remain valid after letting $U \to \infty$.
To evaluate the Hurwitz zeta function, we set
\begin{equation*}
f(k) = \frac{1}{(a + k)^s} = \exp(-s \log(a+k))
\end{equation*}
with the conventional logarithm branch cut on $(-\infty,0)$.
The derivatives of $f(k)$ are given by
\begin{equation*}
f^{(r)}(k) = \frac{(-1)^r (s)_{r}}{(a+k)^{s+r}}
\end{equation*}
where $(s)_{r} = s (s+1) \cdots (s+r-1)$ denotes a rising factorial.
The Euler-Maclaurin summation formula now gives,
at least for $\Re(s) > 1$ and $a \ne 0, -1, -2, \ldots$,
\begin{equation}
\zeta(s,a) = \sum_{k=0}^{N-1} f(k) + \sum_{k=N}^{\infty} f(k) =
S + I + T + R
\label{eq:emzeta}
\end{equation}
where
\begin{align}
S & = \sum_{k=0}^{N-1} \frac{1}{(a+k)^s}, \label{eq:powersumS} \\
I & = \int_N^{\infty} \frac{1}{(a + t)^s} dt = \frac{(a+N)^{1-s}}{s-1}, \\
T & = \frac{1}{(a+N)^s} \left( \frac{1}{2} + \sum_{k=1}^M \frac{B_{2k}}{(2k)!} \frac{(s)_{2k-1}}{(a+N)^{2k-1}} \right), \\
R & = -\int_N^{\infty} \frac{\tilde B_{2M}(t)}{(2M)!} \frac{(s)_{2M}}{(a+t)^{s+2M}} dt. \label{eq:emzetaparts}
\end{align}
If we choose $N$ and $M$ such that $\Re(a+N) > 0$ and $\Re(s+2M-1) > 0$,
the integrals in $I$ and $R$ are well-defined,
giving us the analytic continuation of $\zeta(s,a)$ to
$s \in \mathbb{C}$ except for the pole at $s = 1$. 

In order to evaluate derivatives with respect to $s$ of $\zeta(s,a)$, we
substitute ${s \to s + x \in \mathbb{C}[[x]]}$ and 
evaluate \eqref{eq:emzeta}--\eqref{eq:emzetaparts}
with the corresponding arithmetic operations done
on formal power series (which may be truncated
at some arbitrary finite order in an implementation).
For example, the summand in \eqref{eq:powersumS} becomes
\begin{equation}
\frac{1}{(a+k)^{s+x}} = \sum_{i=0}^{\infty} \frac{(-1)^i \log(a+k)^i}{(a+k)^s} \, x^i \in \mathbb{C}[[x]].
\label{eq:powseries}
\end{equation}
Note that we can evaluate $\zeta(S,a)$ for any formal power series
$S = s + s_1 x + s_2 x^2 + \ldots$ by first evaluating $\zeta(s+x,a)$
and then formally right-composing by $S - s$.
We can also easily evaluate derivatives of $\zeta(s,a) - 1/(s-1)$ at $s = 1$.
The pole of $\zeta(s,a)$ only appears in the term $I$ on the
right hand side of \eqref{eq:emzeta}, so we can remove the singularity
as
\begin{align}
\notag \lim_{s \to 1} \left[ I - \frac{1}{(s+x)-1} = \frac{(a+N)^{1-(s+x)}}{(s+x)-1} - \frac{1}{(s+x)-1} \right] \\
\; = \sum_{i=0}^{\infty} \frac{(-1)^{i+1} \log(a+N)^{i+1}}{i!} \; x^i \; \in \mathbb{C}[[x]].
\label{eq:slimit}
\end{align}

For $F = \sum_k f_k x^k \in \mathbb{C}[[x]]$, define
$|F| = \sum_k |f_k| x^k \in \mathbb{R}[[x]]$.
If it holds for all $k$ that $|f_k| \le |g_k|$,
we write $|F| \le |G|$.
Clearly $|F + G| \le |F| + |G|$ and $|FG| \le |F||G|$.
With this notation, we wish to bound $|R(s+x)|$ where
$R(s) = R$ is the remainder integral given in \eqref{eq:emzetaparts}.

To express the error bound in a compact form, we introduce the
sequence of integrals defined for integers $k \ge 0$
and real parameters $A > 0, B > 1, C \ge 0$ by
\begin{equation*}
J_k(A,B,C) \equiv \int_A^{\infty} t^{-B} (C + \log t)^k dt.
\end{equation*}
Using the binomial theorem, $J_k(A,B,C)$ can be evaluated in closed
form for any fixed $k$. In fact, collecting factors gives
\begin{equation*}
J_k(A,B,C) = \frac{L_k}{(B-1)^{k+1} A^{B-1}}
\end{equation*}
where $L_0 = 1$, $L_k = k L_{k-1} + D^k$ and $D = (B-1) (C + \log A)$.
This recurrence allows computing $J_0, J_1, \ldots, J_n$ easily, using $O(n)$
arithmetic operations.

\begin{theorem}
\label{thm:bound}
Given complex numbers $s = \sigma + \tau i$, $a = \alpha + \beta i$
and positive integers $N, M$ such that $\alpha + N > 1$ and $\sigma + 2M > 1$,
the error term \eqref{eq:emzetaparts}
in the Euler-Maclaurin summation formula applied
to $\zeta(s+x,a) \in \mathbb{C}[[x]]$ satisfies
\begin{equation}
|R(s+x)| \le \frac{4 \left| (s+x)_{2M} \right|}{(2 \pi)^{2M}} \left| \sum_{k=0}^{\infty} R_k x^k \right| \in \mathbb{R}[[x]]
\label{eq:mainbound}
\end{equation}
where $R_k \le (K / k!) \, J_k(N+\alpha, \sigma + 2M, C)$, with
\begin{align}
\label{eq:Cdef}
C & = \frac{1}{2}\log\left(1+\frac{\beta^2}{(\alpha+N)^2}\right) +
    \operatorname{atan}\left(\frac{|\beta|}{\alpha+N}\right)
\end{align}
and
\begin{equation}
\label{eq:Kdef}
K = \exp\left(\max\left(0, \tau \operatorname{atan}\left(\frac{\beta}{\alpha + N}\right)\right)\right).
\end{equation}
\end{theorem}

\begin{proof}
We have
\begin{align*}
|R(s+x)|
&= \left|\int_N^{\infty} \frac{\tilde B_{2M}(t)}{(2M)!}
\frac{(s+x)_{2M}}{(a+t)^{s+x+2M}} dt\right| \\
&\le \int_N^{\infty} \left| \frac{\tilde B_{2M}(t)}{(2M)!}
\frac{(s+x)_{2M}}{(a+t)^{s+x+2M}} \right| dt \\
&\le \frac{4 \left| (s+x)_{2M} \right|}{(2 \pi)^{2M}}
    \int_N^{\infty} \left| \frac{1}{(a+t)^{s+x+2M}} \right| dt
\end{align*}
where the last step invokes the fact that
\begin{equation*}
|\tilde B_{2M}(t)| < \frac{4 (2M)!}{(2\pi)^{2M}}.
\end{equation*}
Thus it remains to bound the coefficients $R_k$ satisfying
\begin{equation*}
\int_N^{\infty} \left| \frac{1}{(a+t)^{s+x+2M}} \right| dt = 
    \sum_k R_k x^k, \quad
    R_k = \int_N^{\infty} \frac{1}{k!}
    \left| \frac{\log(a+t)^k}{(a+t)^{s+2M}} \right| dt.
\end{equation*}
By the assumption that $\alpha + t \ge \alpha + N \ge 1$, we have
\begin{align*}
|\log(\alpha + \beta i + t)|
&= \left|\log(\alpha + t) +
        \log\left( 1 + \frac{\beta i}{\alpha + t}\right) \right| \\
&\le \log(\alpha + t) + \left|\log\left(1 + \frac{\beta i}{\alpha+t}\right)\right| \\
&= \log(\alpha+t) + \left|\frac{1}{2}\log\left(1+\frac{\beta^2}{(\alpha+t)^2}\right)
    + i \operatorname{atan}\left(\frac{\beta}{\alpha + t}\right)\right| \\
&\le \log(\alpha + t) + C
\end{align*}
where $C$ is defined as in \eqref{eq:Cdef}. By the assumption that $\sigma + 2M > 1$, we have
\begin{equation*}
\frac{1}{|(\alpha+\beta i+t)^{\sigma+\tau i + 2M}|}
= \frac{\exp(\tau \operatorname{arg}(\alpha+\beta i+t))}{|\alpha+\beta i+t|^{\sigma+2M}}
\le \frac{K}{(\alpha + t)^{\sigma+2M}}
\end{equation*}
where $K$ is defined as in \eqref{eq:Kdef}. Bounding the
integrand in $R_k$ in terms of the integrand in the definition
of $J_k$ now gives the result.
\end{proof}

The bound given in Theorem \ref{thm:bound} should generally approximate
the exact remainder \eqref{eq:emzetaparts} quite well,
even for derivatives of large order,
if $|a|$ is not too large. The quantity $K$ is especially crude,
however, as it does not decrease when $|a+t|^{-\tau i}$
decreases exponentially as a function of $\tau$.
We have made this simplification in order to obtain
a bound that is easy to evaluate for all $s, a$.
In fact, assuming that $a$ is small, we can simplify the bounds
a bit further using
\begin{equation*}
C \le \frac{\beta^2}{2 (\alpha+N)^2} + \frac{|\beta|}{(\alpha+N)}.
\end{equation*}
In practice, the Hurwitz zeta function is usually only
considered for $0 < a \le 1$, unless $s$ is an integer greater than 1
in which case it reduces to a polygamma function of $a$.
It is easy to derive error bounds for polygamma functions
that are accurate for large $|a|$, and we do not consider
this special case further here.

\section{Algorithmic matters}

The evaluation of $\zeta(s+x,a)$ can be broken into three stages:

\begin{enumerate}
\item Choosing parameters $M$ and $N$ and bounding the remainder $R$.
\item Evaluating the power sum $S$.
\item Evaluating the tail $T$ (and the trivial term $I$).
\end{enumerate}
In this section, we describe some algorithmic techniques that are
useful at each stage. We sketch the computational complexities, but
do not attempt to prove strict complexity bounds.

We assume that arithmetic on real and complex numbers is done using
ball arithmetic~\cite{vdH:ball}, which essentially is floating-point
arithmetic with the added automatic propagation of error bounds. This is
probably the most reasonable approach: \emph{a priori} floating-point error
analysis would be overwhelming to do in full generality (an analysis of the
floating-point error when evaluating $\zeta(s)$ for real $s$, with a partial
analysis of the complex case, is given in \cite{petermann:inria-00070174}).

Using algorithms based on the Fast Fourier Transform (FFT), arithmetic
operations on $b$-bit approximations of real or complex numbers can be done
in time $\tilde O(b)$, where the $\tilde O$-notation suppresses logarithmic
factors. This estimate also holds for division and evaluation of
elementary functions.

Likewise, polynomials of degree $n$ can be multiplied in
$\tilde O(n)$ coefficient operations. Here some care is required: when doing
arithmetic with polynomials that have approximate coefficients,
the \emph{accuracy} of the result can be much lower than the working precision,
depending on the shape of the polynomials and the multiplication algorithm.
If the coefficients vary in magnitude as $2^{\pm \tilde O(n)}$, we may need
$\tilde O(n)$ bits of precision to get an accurate result, making the
effective complexity $\tilde O(n^2)$. This issue is discussed further in \cite{vdH:stablemult}.

Many operations can be reduced to fast multiplication. In particular, we will need
the \emph{binary splitting} algorithm: if a sequence $c_n$ of integers
(or polynomials) satisfies a suitable linear recurrence relation and its
bit size (or degree) grows as $\tilde O(n)$, then we can use a balanced product
tree to evaluate $c_n$ using $\tilde O(n)$ bit (or coefficient) operations,
versus $\tilde O(n^2)$ for repeated application of the recurrence
relation \cite{Bernstein2008, HaiblePapanikolau1998}.

\subsection{Evaluating the error bound}

For a precision of $P$ bits, we should choose $N \sim M \sim P$.
A simple strategy is to
do a binary search for an $N$ that makes
the error bound small enough when $M = cN$ where $c \approx 1$.
This is sufficient for our present purposes, but more
sophisticated approaches are possible.
In particular, for evaluation at large heights
in the critical strip, $N$ should be larger than $M$.

Given complex balls for $s$ and $a$, and
integers $N$ and $M$,
we can evaluate the error bound \eqref{eq:mainbound} using
ball arithmetic.
The output is a power series
with ball coefficients.
The absolute value of each coefficient in this series
should be added to the radius for the
corresponding coefficient in $S + I + T \approx \zeta(s+x,a)$
at the end of the whole computation.
If the assumptions that
$\Re(a) + N > 1$ and $\Re(s) + 2M > 1$
are not satisfied for all points in the balls $s$ and $a$,
we set the error bounds for all coefficients to $+\infty$.

If we are computing $D$ derivatives and $D$ is large,
the rising factorial $|(s+x)_{2M}|$ can be computed
using binary splitting and the outer power series product in 
\eqref{eq:mainbound} can be done using
fast polynomial multiplication, so that only $\tilde O(D + M)$
real number operations are required.
Or, if $D$ is small and $M$ is large, $|(s+x)_{2M}|$ can
be computed via the gamma function in time independent of $M$

\subsection{Evaluating the power sum}

As a power series, the power sum $S$ becomes $\sum_{k=0}^{N-1} (\sum_i c_i(k) x^i)$
where the coefficients $c_i(k)$ are given by \eqref{eq:powseries}.
For $i \ge 1$, the coefficients can be computed using the recurrence
\begin{align*}
c_{i+1}(k) = -\frac{\log(a+k)}{i+1} c_i(k).
\end{align*}
If we are computing $D$ derivatives with a working precision of $P$ bits, the complexity
of evaluating the power sum is $\tilde O(NPD)$, or $\tilde O(N^2 D)$ if $N \sim P$.
The computation is easy to parallelize
by assigning a range of $k$ values to each thread
(for large $D$, a more memory-efficient method is to assign a range of
$i$ to each thread).

\begin{algorithm}
  \caption{Sieved summation of a completely multiplicative function}
  \begin{algorithmic}[1]
  \Require A function $f$ such that $f(jk) = f(j) f(k)$ for $j,k \in \mathbb{Z}_{\ge 1}$, and an integer $N \ge 1$
  \Ensure{$\sum_{k=1}^N f(k)$}
  \State $p \gets 2^{\lfloor \log_2 N \rfloor}$ (largest power of two such that $p \le N$)
  \State $h \gets 1$, $z \gets 0$, $u \gets 0$
  \State $D = [\,]$ \Comment{Build table of divisors}
  \For {$k \gets 1; \; k \le N; \; k \gets k + 2$}
    \State $D[k] \gets 0$
  \EndFor
  \For {$k \gets 3; \; k \le \lfloor \sqrt N \rfloor; \; k \gets k + 2$}
    \If {$D[k] = 0$}
      \For {$j \gets k^2; \; j \le N; \; j \gets j + 2k$}
        \State $D[j] \gets k$
      \EndFor
    \EndIf
  \EndFor
  \State $F = [\,]$ \Comment{Create initially empty cache of $f(k)$ values}
  \State $F[2] \gets f(2)$
  \For {$k \gets 1; \; k \le N; \; k \gets k + 2$}
    \If {$D[k] = 0$} \Comment{$k$ is prime (or 1)}
      \State $t \gets f(k)$
    \Else
      \State $t \gets F[D[k]] F[k / D[k]]$ \Comment{$k$ is composite}
    \EndIf
    \If {$3k \le N$}
      \State $F[k] \gets t$ \Comment{Store $f(k)$ for future use}
    \EndIf
    \State $u \gets u + t$
    \While {$k = h$ and $p \ne 1$} \Comment{Horner's rule}
      \State $z \gets u + F[2] z$
      \State $p \gets p / 2$
      \State $h \gets \lfloor N / p \rfloor$
      \If {$h$ is even} \State $h \gets h - 1$ \EndIf
    \EndWhile
  \EndFor
  \State \Return $u + F[2] z$
  \end{algorithmic}
  \label{alg:multsum}
\end{algorithm}

When evaluating the ordinary Riemann zeta function, i.e. when $a = 1$,
and we just want to compute a small number of derivatives,
we can speed up the power sum a bit. Writing the sum as
$\sum_{k=1}^N f(k)$, the
terms $f(k) = k^{-(s+x)}$
are completely multiplicative, i.e. $f(k_1 k_2) = f(k_1) f(k_2)$.
This means that we only need to evaluate $f(k)$ from scratch
when $k$ is prime; when $k$ is composite, a single multiplication
is sufficient.

This method has two drawbacks: we have to store
previously computed terms, which requires $O(NPD)$ space,
and the power series multiplication $f(k_1) f(k_2)$ becomes
more expensive than evaluating $f(k_1 k_2)$
from scratch for large $D$. For both reasons,
this method is only useful when $D$ is quite small (say $D \le 4$).

We can avoid some redundant work by collecting
multiples of small primes. For example, if we extract all powers of two,
$\sum_{k=1}^{10} f(k)$ can be written as
\begin{align*}
       & \; [f(1) + f(3) + f(5) + f(7) + f(9)] \\
+ f(2) & \; [f(1) + f(3) + f(5)]\\
+ f(4) & \; [f(1)] \\
+ f(8) & \; [f(1)].
\end{align*}
This is a polynomial in $f(2)$ and can be evaluated from bottom
to top using Horner's rule while progressively adding the terms
in the brackets.
Asymptotically, this reduces the number of multiplications
and the size of the tables by half.
Algorithm \ref{alg:multsum} implements this trick, and requires about
$\pi(N) \approx N / \log N$ evaluations of $f(k)$ and
$N / 2$ multiplications,
at the expense of having to store about $N / 6$ function values
plus a table of divisors of about $N / 2$ integers.
Constructing the table of divisors using the sieve of
Eratosthenes requires $O(N \log \log N)$ integer operations, but
this cost is negligible when multiplications and $f(k)$ evaluations
are expensive.
One could also extract other powers besides 2
(for example powers of 3 and 5), but this gives diminishing returns.

Another trick that can save time at high precision is to avoid
computing the logarithms of integers from scratch.
If $q$ and $p$ are nearby integers (such as two consecutive primes)
and we already know $\log(p)$, we can use the identity
\begin{equation*}
\log(q) = \log(p) + 2 \operatorname{atanh} \left( \frac{q-p}{q+p} \right)
\end{equation*}
and evaluate the inverse hyperbolic tangent by applying
binary splitting to its Taylor series. This is not
an asymptotic improvement over the best known algorithm
for computing the logarithm (which uses the arithmetic-geometric mean),
but likely faster in practice.

\subsection{Evaluating the tail}

Except for the multiplication by Bernoulli numbers, the terms of the tail sum $T$
satisfy a simple (hypergeometric) recurrence relation.
If we are computing $D$ derivatives with a working precision of $P$ bits, the complexity
of evaluating the tail by repeated application of the recurrence relation
is $\tilde O(MPD)$, or $\tilde O(P^2 D)$ if $M \sim P$.
We can do better if $D$ is large, using binary splitting (Algorithm~\ref{alg:bsplittail}).

\begin{algorithm}
  \caption{Evaluation of the tail $T$ using binary splitting}
  \label{alg:rsps}
  \begin{algorithmic}[1]
    \Require $s, a \in \mathbb{C}$ and $N, M, D \in \mathbb{Z}_{\ge 1}$
    \Ensure{$T = \displaystyle{\frac{1}{(a+N)^{s+x}}
        \left( \frac{1}{2} + \sum_{k=1}^M \frac{B_{2k}}{(2k)!}
        \frac{(s+x)_{2k-1}}{(a+N)^{2k-1}} \right)} \in \mathbb{C}[[x]] / \langle x^D \rangle$}
    \State Let $x$ denote the generator of $\mathbb{C}[[x]] / \langle x^D \rangle$
    \Function{BinSplit}{$j$, $k$}
        \If {$j + 1 = k$}
            \If {$j = 0$}
                \State{$\displaystyle{P \gets (s + x)/(2(a+N))}$}
            \Else
                \State{$\displaystyle{P \gets \frac{(s+2j-1+x)(s+2j+x)}{(2j+1)(2j+2) (a+N)^2}}$}
            \EndIf
            \State \Return $(P, \, B_{2j+2} P)$
        \Else
            \State $(P_1, \, R_1) \gets \Call{BinSplit}{j, \, \lfloor (j+k)/2 \rfloor}$
            \State $(P_2, \, R_2) \gets \Call{BinSplit}{\lfloor (j+k)/2 \rfloor, \, k}$
            \State \Return $(P_1 P_2, \, R_1 + P_1 R_2)$ \Comment{Polynomial multiplications mod $x^D$}
        \EndIf
    \EndFunction
    \State $(P, \, R) \gets \Call{BinSplit}{0, M}$
    \State $T \gets (a+N)^{-(s+x)} (1/2 + R)$ \Comment{Polynomial multiplication mod $x^D$}
    \State \Return {$T$}
  \end{algorithmic}
  \label{alg:bsplittail}
\end{algorithm}

If $D \sim M$, the complexity with binary splitting is only $\tilde O(PD)$,
or softly optimal in the bit size of the output.
A drawback is that the intermediate products increase the memory consumption.

The Bernoulli numbers can of course be cached for repeated
evaluation of the zeta function, but computing them the first
time can be a bottleneck at high precision, at least if done naively.
The first~$2M$ Bernoulli numbers
can be computed in quasi-optimal time $\tilde O(M^2)$,
for example by using Newton iteration and fast polynomial multiplication
to invert the power series $(e^x - 1) / x$.
For most practical purposes, simpler algorithms
with a time complexity of $\tilde O(M^3)$ are adequate, however.
Various algorithms are discussed in \cite{HarveyBrent2011}.
An interesting alternative, used in unpublished work of
Bloemen \cite{Bloemen2009}, is to compute $B_n$ via $\zeta(n)$
by direct approximation of the sum $\sum_{k=0}^{\infty} k^{-n}$,
recycling the powers to process several $n$ simultaneously.

\section{Implementation and benchmarks}

We have implemented the Hurwitz zeta function for
$s \in \mathbb{C}[[x]]$ and $a \in \mathbb{C}$ with rigorous error bounds as
part of the Arb library\footnote{\url{http://fredrikj.net/arb}}. This library
is written in C and is freely available under version~2 or later of the
GNU General Public License. It uses the MPFR \cite{Fousse2007} library for
evaluation of some elementary functions, GMP \cite{GMP} or MPIR \cite{MPIR}
for integer arithmetic, and FLINT \cite{Hart2010} for polynomial arithmetic.

Our implementation incorporates most of the techniques discussed in the previous section,
including optional parallelization of the power sum. Bernoulli numbers are
computed using the algorithm of Bloemen. Fast and numerically stable
multiplication in $\mathbb{R}[x]$ and $\mathbb{C}[x]$ is implemented by
rescaling polynomials and breaking them into segments with similarly-sized
coefficients and computing the subproducts exactly in $\mathbb{Z}[x]$
(a simplified version of van der Hoeven's block multiplication algorithm \cite{vdH:stablemult}).
Polynomial multiplication in $\mathbb{Z}[x]$ is done via FLINT which
for large polynomials uses a Sch\"{o}nhage-Strassen FFT implementation by William Hart.

\subsection{Computing zeros to high precision}

For $n \ge 1$, let $\rho_n$ denote the $n$-th smallest zero of $\zeta(s)$ with
positive imaginary part. We assume that $\rho_n$ is simple and has real part $1/2$.
Using Newton's method, we can evaluate $\rho_n$ to high precision nearly as
fast as we can evaluate $\zeta(s)$ for $s$ near $\rho_n$.

It is convenient to work with real numbers. The ordinate $t_n = \Im(\rho_n)$ is
a simple zero of the real-valued function $Z(t) = e^{i \theta(t)} \zeta(1/2+it)$ where
\begin{equation*}
\theta(t) = \frac{\log\Gamma\left(\frac{2it+1}{4}\right)-\log\Gamma\left(\frac{-2it+1}{4}\right)}{2i} - \frac{\log \pi}{2} t.
\end{equation*}
We assume that we are given an isolating ball $B_0 = [m_0-\varepsilon_0, m_0+\varepsilon_0]$
such that $t_n \in B_0$ and ${t_{m} \not \in B_0}, {m \ne n}$, and wish to
compute $t_n$ to high precision (finding such a ball for a given $n$ is an
interesting problem, but we do not consider it here).

Newton's method maps an approximation $z_n$ of a root of a real analytic
function $f(z)$ to a new approximation $z_{n+1}$ via $z_{n+1} = z_n - f(z_n) / f'(z_n)$.
Using Taylor's theorem, the error can be shown to satisfy
\begin{equation*}
\left| {\epsilon_{n+1}}\right| = \frac {\left| f^{\prime\prime} (\xi_n) \right| }{2 \left| f^\prime(z_n) \right|} \, |\epsilon_n|^2
\end{equation*}
for some $\xi_n$ between $z_n$ and the root.

As a setup step, we evaluate $Z(s), Z'(s), Z''(s)$
(simultaneously using power series arithmetic) at $s = B_0$, and compute
\begin{equation*}
C = \frac{\max |Z''(B_0)|}{2 \min |Z'(B_0)|}.
\end{equation*}
This only needs to be done at low precision.

Starting from an input ball $B_k = [m_k-\varepsilon_k, m_k+\varepsilon_k]$,
one step with Newton's method gives an output ball
$B_{k+1} = [m_{k+1}-\varepsilon_{k+1}, m_{k+1}+\varepsilon_{k+1}]$.
The updated midpoint is given by
\begin{equation}
m_{k+1} = m_k - \frac{Z(m_k)}{Z'(m_k)}
\label{eq:newtonstep}
\end{equation}
where we evaluate $Z(m_k)$ and $Z'(m_k)$ simultaneously using power series
arithmetic. The updated radius is given by
$\varepsilon_{k+1} = \varepsilon'_{k+1} + C \varepsilon_k^2$ where $\varepsilon'_{k+1}$
is the numerical error (or a bound thereof) resulting from evaluating
\eqref{eq:newtonstep} using finite-precision arithmetic.
The new ball is valid as long as $B_{k+1} \subseteq B_k$ (if this does not hold,
the algorithm fails and we need to start with a better $B_0$ or
increase the working precision).

For best performance, the evaluation precision should be chosen so that
$\varepsilon'_{k+1} \approx C \varepsilon_k^2$. In other words, for a
target accuracy of $p$ bits, the evaluations should be done at
$\ldots, p/4, p/2, p$ bits, plus some guard bits.

As a benchmark problem, we compute an approximation $\tilde \rho_1$ of the
first nontrivial zero $\rho_1 \approx 1/2 + 14.1347251417i$ and then evaluate
$\zeta(\tilde \rho_1)$ to the same precision.  We compare our implementation
of the zeta function and the root-refinement algorithm described above
(starting from a double-precision isolating ball) with the
\texttt{zetazero} and \texttt{zeta} functions provided in mpmath version 0.17 in
Sage 5.10 \cite{sage} and the \texttt{ZetaZero} and \texttt{Zeta} functions provided
in Mathematica 9.0. The results of this benchmark are shown
in Table \ref{tab:zerotimings}. At 10000 digits, our code for computing the
zero is about two orders of magnitude faster than the other systems, and the
subsequent single zeta evaluation is about one order of magnitude faster.

We have computed $\rho_1$ to 303000 digits, or slightly more than one
million bits, which appears to be a record (a 20000-digit value is given in
\cite{MatiyasevichBeliakov2011}). The computation used up to 62 GiB of memory
for the sieved power sum and the storage of Bernoulli numbers up to $B_{325328}$
(to attain even higher precision, the memory usage could be reduced by
evaluating the power sum without sieving, perhaps using several CPUs in
parallel, and not caching Bernoulli numbers).

\begin{table}
\centering
\begin{tabular}{|l | r r | r r | r r |}
\hline
Digits
& \multicolumn{2}{|c|}{mpmath}
& \multicolumn{2}{|c|}{Mathematica}
& \multicolumn{2}{|c|}{Arb} \\
\hline
    & $\tilde \rho_1$ & $\zeta(\tilde \rho_1)$
    & $\tilde \rho_1$ & $\zeta(\tilde \rho_1)$ & $\tilde \rho_1$ & $\zeta(\tilde \rho_1)$ \\
100    & 0.080 & 0.0031 & 0.044 & 0.012 & 0.012 & 0.0011 \\
1000   & 7.1 & 0.24 & 11 & 1.6 & 0.18 & 0.05 \\
10000  & 7035 & 252 & 5127 & 779 & 29 & 15 \\
100000 & - & - & - & - & 6930   & 3476 \\
303000 & - & - & - & - & 73225 & 31772 \\
\hline
\end{tabular}
\caption{Time in seconds to compute an approximation $\tilde \rho_1$ of the
first nontrivial zero $\rho_1$ accurate to the specified number of decimal digits,
and then to evaluate $\zeta(\tilde \rho_1)$ at the same precision.
Computations were done on a 64-bit Intel Xeon E5-2650 2.00 GHz CPU.}
\label{tab:zerotimings}
\end{table}

\subsection{Computing the Keiper-Li coefficients}

Riemann's function $\xi(s) = \frac{1}{2} s (s-1) \pi^{-s/2} \Gamma(s/2) \zeta(s)$
satisfies the symmetric functional equation $\xi(s) = \xi(1-s)$.
The coefficients $\{\lambda_n\}_{n=1}^{\infty}$ defined by
\begin{equation*}
\log \xi\left(\frac{1}{1-x}\right) = \log \xi\left(\frac{x}{x-1}\right) = -\log 2 + \sum_{n=1}^{\infty} \lambda_n x^n
\end{equation*}
were introduced by Keiper \cite{Keiper1992power}, who
noted that the truth of the Riemann hypothesis would imply
that $\lambda_n > 0$ for all $n > 0$.
In fact, Keiper observed that if one makes an assumption about the
distribution of the zeros of $\zeta(s)$
that is even stronger than the Riemann hypothesis, the coefficients
$\lambda_n$ should behave as
\begin{equation}
\lambda_n \approx (1/2) \left(\log n - \log(2\pi) + \gamma - 1\right).
\label{eq:liapprox}
\end{equation}
Keiper presented numerical evidence for this conjecture by
computing $\lambda_n$ up to $n = 7000$, showing that the
approximation error appears to fluctuate increasingly close to zero.
Some years later, Li proved \cite{Li1997positivity} that the Riemann hypothesis
actually is equivalent to the positivity of $\lambda_n$ for all $n > 0$
(this reformulation of the Riemann hypothesis is known as Li's criterion).
Recently, Arias de Reyna has proved that a certain precise statement
of \eqref{eq:liapprox} also is equivalent to
the Riemann hypothesis~\cite{Arias2011}.

\begin{table}
\centering
\begin{tabular}{|l | r r r|}
\hline
                      &  $n = 1000$   &  $n = 10000$   & $n = 100000$  \\
\hline
1: Error bound           &    0.017      &   1.0    &    97               \\
1: Power sum             &    0.048      &   47     &     65402   \\
(1: Power sum, CPU time) &    (0.65)     &   (693)  &    (1042210) \\
1: Bernoulli numbers     &    0.0020     &   0.19  &     59               \\
1: Tail                  &    0.058      &   11    &     1972             \\
\hline
2: Series logarithm      &    0.047      &   8.5   &     1126             \\
3: log $\Gamma(1+x)$ series &    0.019      &   3.0   &  1610             \\
4: Composition           &    0.022      &   4.1   &     593              \\
\hline
Total wall time          &    0.23       &   84   &      71051  \\
\hline
Peak RAM usage (MiB)   &    8        &   730   &     48700            \\
\hline
\end{tabular}
\caption{Elapsed time in seconds to evaluate the Keiper-Li coefficients
$\lambda_0 \ldots \lambda_n$ with a working precision of $1.1n + 50$ bits,
giving roughly $0.1n$ accurate bits. The computations were done
on a multicore system with 64-bit Intel Xeon E7-8837 2.67 GHz CPUs
(16 threads were used for the power sum, and all
other parts were computed serially on a single core).}
\label{tab:listats}
\end{table}

\begin{figure}[width=13cm]
\begin{center}
\includegraphics[width=12cm]{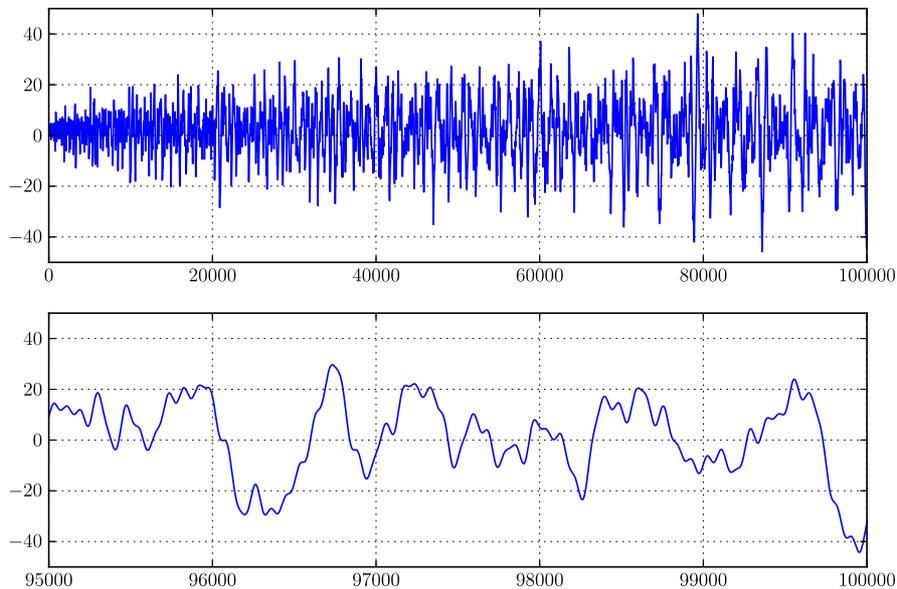}
\caption{Plot of $n \left( \lambda_n - (\log n - \log(2\pi) + \gamma - 1)/2\right)$.}
\label{plt:liplot}
\end{center}
\end{figure}

A computation of the Keiper-Li coefficients up to $n = 100000$
shows agreement with Keiper's conjecture (and the Riemann hypothesis),
as illustrated in Figure \ref{plt:liplot}.
We obtain
$\lambda_{100000} = 4.62580782406902231409416038\ldots$
(plus about 2900 more accurate digits),
whereas \eqref{eq:liapprox} gives $\lambda_{100000} \approx 4.626132$.
Empirically, we need a working precision of about $n$ bits
to determine $\lambda_n$ accurately.
A breakdown of the computation time to determine
the signs of $\lambda_n$ up to $n = 1000$, $10000$ and $100000$
is shown in Table \ref{tab:listats}.


Our computation of the Keiper-Li coefficients uses the formula
\begin{equation*}
\log \xi(s) = \log(-\zeta(s)) + \log \Gamma\left(1+\frac{s}{2}\right)
+ \log(1-s) - \frac{s \log \pi}{2}
\end{equation*}
which we evaluate at $s = x \in \mathbb{R}[[x]]$. This arrangement of
the terms avoids singularities and branch cuts at the expansion point.
We carry out the following steps (plus some more trivial operations):

\begin{enumerate}
\item Computing the series expansion of $\zeta(s)$ at $s = 0$.
\item Computing the logarithm of a power series, i.e. $\log f(x) = \int f'(x) / f(x) dx$.
\item Computing the series expansion of $\log \Gamma(s)$ at $s = 1$, i.e. computing $\gamma, \zeta(2), \zeta(3), \zeta(4), \ldots$.
\item Finally, right-composing by $x / (x - 1)$ to obtain the Keiper-Li coefficients.
\end{enumerate}

Step 2 requires $O(M(n))$ arithmetic operations on real numbers.
We use a hybrid algorithm to compute the integer zeta values in step 3;
the details are beyond the scope of the present paper.

There is a very fast way to perform step 4. For
$f = \sum_{k=0}^{\infty} a_k x^k \in \mathbb{C}[[x]]$, the binomial (or Euler)
transform $T \colon \mathbb{C}[[x]] \to \mathbb{C}[[x]]$ is defined by
\begin{equation*}
T[f(x)] = \frac{1}{1-x} f\left(\frac{x}{x-1}\right) =
\sum_{n=0}^{\infty} \left(
\sum_{k=0}^n (-1)^k {n \choose k} a_k
\right) x^n.
\end{equation*}
We have
\begin{equation*}
f\left(\frac{x}{x-1}\right) = a_0 + x T\left[\frac{a_0 - f}{x}\right].
\end{equation*}
If $B \colon \mathbb{C}[[x]] \to \mathbb{C}[[x]]$ denotes the Borel transform
\begin{equation*}
B \left[\sum_{k=0}^{\infty} a_k x^k\right] = \sum_{k=0}^{\infty} \frac{a_k}{k!} x^k,
\end{equation*}
then (see \cite{Gould1990}) $T[f(x)] = B^{-1}[e^x B[f(-x)]]$.
This identity gives an algorithm for evaluating the composition which
requires only $M(n) + O(n)$ coefficient operations
where $M(n) = \tilde O(n)$ is the operation complexity of polynomial
multiplication.
Moreover, this algorithm is numerically stable (in the sense that it does not
significantly increase errors from the input when using ball arithmetic),
provided that a numerically stable polynomial multiplication algorithm is used.

The composition could also be carried out using various generic algorithms for
composition of power series. We tested three other algorithms,
and found them to perform much worse:

\begin{itemize}
\item Horner's rule is slow (requiring about $n M(n)$ operations) and is
numerically unsatisfactory in the sense that it gives extremely poor error bounds with ball arithmetic.
\item The Brent-Kung algorithm based on matrix multiplication \cite{BrentKung1978}
turns out to give adequate error bounds, but uses about $O(n^{1/2} M(n) + n^2)$
operations which still is expensive for large~$n$.
\item We also tried binary splitting: to evaluate $f(p/q)$ where $f$ is a power
series and $p$ and $q$ are polynomials, we recursively split the evaluation in
half and keep numerator and denominator polynomials separated. In the end, we
perform a single power series division. This only costs $O(M(n) \log n)$
operations, but turns out to be numerically unstable. It would be of
independent interest to investigate whether this algorithm can be modified to avoid the stability problem.
\end{itemize}

\subsection{Computing the Stieltjes constants}

The generalized Stieltjes constants $\gamma_n(a)$ are defined by
\begin{equation*}
\zeta(s,a)=\frac{1}{s-1}+\sum_{n=0}^\infty \frac{(-1)^n}{n!} \gamma_n(a) \; (s-1)^n.
\end{equation*}
The ``usual'' Stieltjes constants are $\gamma_n(1) = \gamma_n$, and
$\gamma_0 = \gamma \approx 0.577216$ is Euler's constant. The Stieltjes
constants were first studied over a century ago. Some historical notes and
numerical values of $\gamma_n$ for $n \le 20$ are given in \cite{Bohman1988}.
Keiper \cite{Keiper1992power} provides a method for computing the Stieltjes
constants based on numerical integration and recurrence relations, and lists
various $\gamma_n$ up to $n = 150$. Keiper's algorithm is implemented in
Mathematica \cite{WolframNotes}.

More recently, Kreminski \cite{Kreminski2003} has given an algorithm for
the Stieltjes constants, also based on numerical integration but different
from Keiper's. He reports having computed $\gamma_n$ to a few thousand digits
for all $n \le 10000$, and provides further isolated values up to
$\gamma_{50000}$ (accurate to 1000 digits) as well as tables of $\gamma_n(a)$
with various $a \ne 1$.

The best proven bounds for the Stieltjes constants appear to be very pessimistic.
In a recent paper, Knessl and Coffey \cite{Knessl2011effective} give an
asymptotic approximation formula for the Stieltjes constants that seems to
be very accurate even for small $n$. Based on numerical computations done with
Mathematica, they note that their approximation correctly predicts the
sign of $\gamma_n$ up to at least $n = 35000$ with the single
exception of $n = 137$.

Our implementation immediately gives the generalized Stieltjes constants
by computing the series expansion of $\zeta(s,a) - 1/(s-1)$ at $s = 1$
using~\eqref{eq:slimit}. The costs are similar to those for computing the
Keiper-Li coefficients: due to ill-conditioning, it appears that we need
about $n + p$ bits of precision to determine $\gamma_n$
with $p$ bits of accuracy. This makes our method somewhat unattractive
for computing just a few digits of $\gamma_n$ when $n$ is large, but
reasonably good if we want a large number of digits. Our method is
also useful if we want to compute a table of all the values
$\gamma_0, \ldots, \gamma_n$ simultaneously.

For example, we can compute $\gamma_n$ for all $n \le 1000$ to 1000-digit
accuracy in just over 10 seconds on a single CPU. Computing the single
coefficient $\gamma_{1000}$ to 1000-digit accuracy with
Mathematica~9.0 takes 80 seconds, with an estimated 20 hours required
for all $n \le 1000$. Thus our implementation is nearly four orders of
magnitude faster. We can compute a table of accurate values of $\gamma_n$
for all $n \le 10000$ in a few minutes on an ordinary workstation
with around one GiB of memory.

We have computed all $\gamma_n$ up to $n = 100000$ using a working
precision of $125050$ bits, resulting in an accuracy from about
37640 decimal digits for $\gamma_0$ to about 10860 accurate digits
for $\gamma_{100000}$. The computation took 26~hours on a multicore
system with 16 threads utilized for the power sum, with a peak memory
consumption of about 80 GiB during the binary splitting evaluation
of the tail. As shown in Figure~\ref{plt:stieltjesplot}, the accuracy
of the Knessl-Coffey approximation approaches six digits on average.
Our computation gives
$\gamma_{100000} = 1.991927306312541095658\ldots \times 10^{83432}$,
while the Knessl-Coffey approximation gives
$\gamma_n \approx 1.9919333 \times 10^{83432}$.
We are able to verify that $n = 137$ is the only instance for $n \le 100000$
where the Knessl-Coffey approximation has the wrong sign.

\begin{figure}[width=13cm]
\begin{center}
\includegraphics[width=12cm]{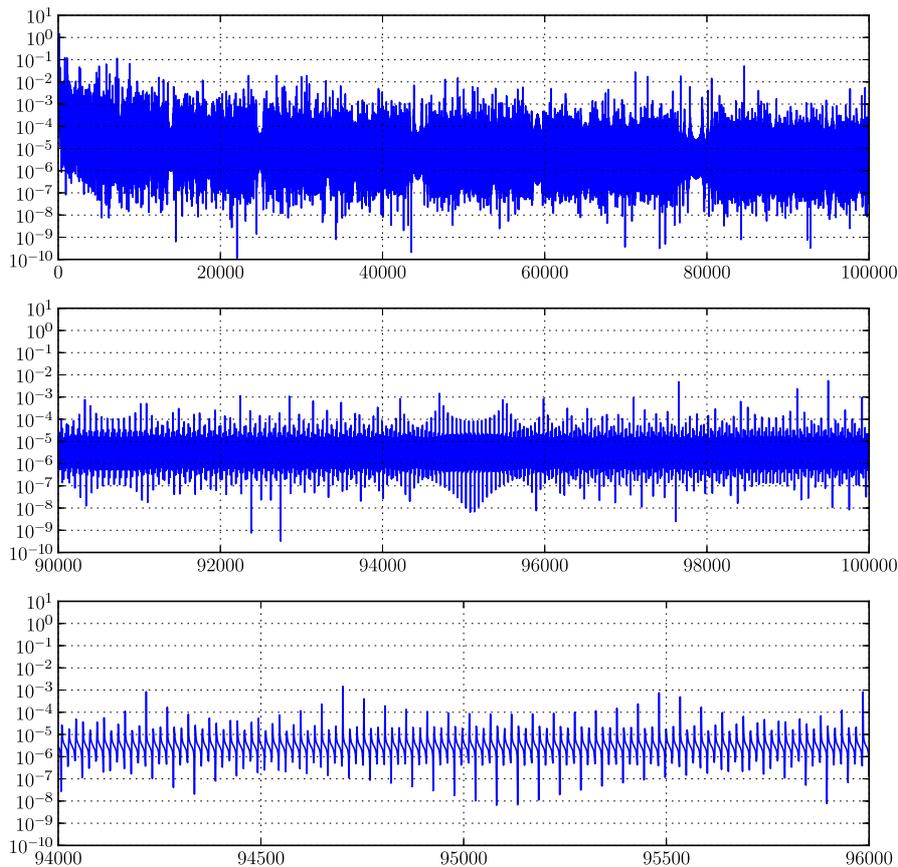}
\caption{Plot of the relative error $|\gamma_n - \tilde \gamma_n|/|\gamma_n|$
of the Knessl-Coffey approximation for the Stieltjes constants.
The error exhibits a complex oscillation pattern.}
\label{plt:stieltjesplot}
\end{center}
\end{figure}

We emphasize that our implementation gives $\gamma_n(a)$ with proved
error bounds, while the other cited works and implementations
(to our knowledge) depend on heuristic error estimates.

We have not yet implemented a function for computing isolated
Stieltjes constants of large index; this would have roughly the same
running time as the evaluation of the tail (since only a single derivative
of the power sum would have to be computed). The memory consumption is
highest when evaluating the tail, and would therefore remain the same.

\section{Discussion}

One direction for further work would be to improve the error bounds for
large $|a|$ and to investigate strategies for selecting $N$ and $M$ optimally,
particularly when the number of derivatives is large.  It would also be
interesting to investigate parallelization of the tail sum, or look for
ways to evaluate a single derivative of high order of the tail in a
memory-efficient way. Further constant-factor improvements are possible
in an implementation, for example by reducing the precision of terms that
have small magnitude (rather than naively performing all operations
at the same precision).

Finally, it would be interesting to compare the efficiency of the
Euler-Maclaurin formula with other approaches to evaluating the
Hurwitz zeta function such as the algorithms of
Borwein \cite{Borwein2000}, Vep{\v{s}}tas \cite{Vepstas2008} and
Coffey \cite{Coffey2009efficient}.

\bibliographystyle{plain}
\bibliography{references.bib}

\begin{thebibliography}{10}

\bibitem{Bailey2000experimental}
D.~H. Bailey and J.~M. Borwein.
\newblock Experimental mathematics: recent developments and future outlook.
\newblock In B.~Engquist, W.~Schmid, and P.~W. Michor, editors, {\em
  Mathematics Unlimited -- 2001 and Beyond}, pages 51--66. Springer, 2000.

\bibitem{Bernstein2008}
D.~J. Bernstein.
\newblock Fast multiplication and its applications.
\newblock {\em Algorithmic Number Theory}, 44:325--384, 2008.

\bibitem{Bloemen2009}
R.~Bloemen.
\newblock Even faster $\zeta(2n)$ calculation!, 2009.
\newblock
  \url{http://remcobloemen.nl/2009/11/even-faster-zeta-calculation.html}.

\bibitem{Bogolubsky2006}
A.~I. Bogolubsky and S.~L. Skorokhodov.
\newblock Fast evaluation of the hypergeometric function ${}_p{F}_{p-1}(a; b;
  z)$ at the singular point $z=1$ by means of the {H}urwitz zeta function
  $\zeta(\alpha, s)$.
\newblock {\em Programming and Computer Software}, 32(3):145--153, 2006.

\bibitem{Bohman1988}
J.~Bohman and C-E. Fr{\"o}berg.
\newblock The {S}tieltjes function -- definition and properties.
\newblock {\em Mathematics of Computation}, 51(183):281--289, 1988.

\bibitem{BorweinBradleyCrandall2000}
J.~M. Borwein, D.~M. Bradley, and R.~E. Crandall.
\newblock Computational strategies for the {R}iemann zeta function.
\newblock {\em Journal of Computational and Applied Mathematics}, 121:247--296,
  2000.

\bibitem{Borwein2000}
P.~Borwein.
\newblock An efficient algorithm for the {R}iemann zeta function.
\newblock {\em Canadian Mathematical Society Conference Proceedings},
  27:29--34, 2000.

\bibitem{BrentKung1978}
R.~P. Brent and H.~T. Kung.
\newblock Fast algorithms for manipulating formal power series.
\newblock {\em Journal of the ACM}, 25(4):581--595, 1978.

\bibitem{Coffey2009efficient}
M.~W. Coffey.
\newblock An efficient algorithm for the {H}urwitz zeta and related functions.
\newblock {\em Journal of Computational and Applied Mathematics},
  225(2):338--346, 2009.

\bibitem{Arias2011}
J.~Arias de~Reyna.
\newblock Asymptotics of {K}eiper-{L}i coefficients.
\newblock {\em Functiones et Approximatio Commentarii Mathematici},
  45(1):7--21, 2011.

\bibitem{GMP}
The~GMP development team.
\newblock {GMP: The {GNU} multiple precision arithmetic library}.
\newblock \url{http://www.gmplib.org}.

\bibitem{MPIR}
The~MPIR development team.
\newblock {MPIR}: {M}ultiple {P}recision {I}ntegers and {R}ationals.
\newblock \url{http://www.mpir.org}.

\bibitem{Edwards1974}
H.~M. Edwards.
\newblock {\em Riemann's zeta function}.
\newblock Academic Press, 1974.

\bibitem{Flajolet1996zeta}
P.~Flajolet and I.~Vardi.
\newblock Zeta function expansions of classical constants.
\newblock Unpublished manuscript,
  \url{http://algo.inria.fr/flajolet/Publications/landau.ps}, 1996.

\bibitem{Fousse2007}
L.~Fousse, G.~Hanrot, V.~Lef\`evre, P.~P\'elissier, and P.~Zimmermann.
\newblock {MPFR}: A multiple-precision binary floating-point library with
  correct rounding.
\newblock {\em {ACM} Transactions on Mathematical Software}, 33(2):13:1--13:15,
  June 2007.
\newblock \url{http://mpfr.org}.

\bibitem{Gould1990}
H.~Gould.
\newblock Series transformations for finding recurrences for sequences.
\newblock {\em Fibonacci Quarterly}, 28:166--171, 1990.

\bibitem{HaiblePapanikolau1998}
B.~Haible and T.~Papanikolaou.
\newblock Fast multiprecision evaluation of series of rational numbers.
\newblock In J.~P. Buhler, editor, {\em Algorithmic Number Theory: Third
  International Symposium}, volume 1423, pages 338--350. Springer, 1998.

\bibitem{Hart2010}
W.~B. Hart.
\newblock Fast {L}ibrary for {N}umber {T}heory: {A}n {I}ntroduction.
\newblock In {\em Proceedings of the Third international congress conference on
  Mathematical software}, ICMS'10, pages 88--91, Berlin, Heidelberg, 2010.
  Springer-Verlag.
\newblock \url{http://flintlib.org}.

\bibitem{HarveyBrent2011}
D.~Harvey and R.~P. Brent.
\newblock Fast computation of {B}ernoulli, tangent and secant numbers, 2011.
\newblock \url{http://arxiv.org/abs/1108.0286}.

\bibitem{Hiary2011fast}
G.~Hiary.
\newblock Fast methods to compute the {R}iemann zeta function.
\newblock {\em Annals of mathematics}, 174:891--946, 2011.

\bibitem{WolframNotes}
Wofram~Research Inc.
\newblock Some notes on internal implementation (section of the online
  documentation for {M}athematica 9.0).
\newblock
  \url{http://reference.wolfram.com/mathematica/tutorial/SomeNotesOnInternalIm%
plementation.html}, 2013.

\bibitem{Keiper1992power}
J.~B. Keiper.
\newblock Power series expansions of {R}iemann's $\xi$ function.
\newblock {\em Mathematics of Computation}, 58(198):765--773, 1992.

\bibitem{Knessl2011effective}
C.~Knessl and M.~Coffey.
\newblock An effective asymptotic formula for the {S}tieltjes constants.
\newblock {\em Mathematics of Computation}, 80(273):379--386, 2011.

\bibitem{Kreminski2003}
R.~Kreminski.
\newblock {N}ewton-{C}otes integration for approximating {S}tieltjes
  (generalized {E}uler) constants.
\newblock {\em Mathematics of Computation}, 72(243):1379--1397, 2003.

\bibitem{Li1997positivity}
X-J. Li.
\newblock The positivity of a sequence of numbers and the {R}iemann
  {H}ypothesis.
\newblock {\em Journal of Number Theory}, 65(2):325--333, 1997.

\bibitem{Matiyasevich2012}
Y.~Matiyasevich.
\newblock An artless method for calculating approximate values of zeros of
  {R}iemann's zeta function, 2012.
\newblock \url{http://logic.pdmi.ras.ru/~yumat/personaljournal/artlessmethod/}.

\bibitem{MatiyasevichBeliakov2011}
Y.~Matiyasevich and G.~Beliakov.
\newblock Zeroes of {R}iemann's zeta function on the critical line with 20000
  decimal digits accuracy, 2011.
\newblock \url{http://dro.deakin.edu.au/view/DU:30051725?print_friendly=true}.

\bibitem{Odlyzko1988fast}
A.~M. Odlyzko and A.~Sch{\"o}nhage.
\newblock Fast algorithms for multiple evaluations of the {R}iemann zeta
  function.
\newblock {\em Transactions of the American Mathematical Society},
  309(2):797--809, 1988.

\bibitem{Olver1997}
F.~W.~J. Olver.
\newblock {\em Asymptotics and Special Functions}.
\newblock A K Peters, Wellesley, MA, 1997.

\bibitem{petermann:inria-00070174}
Y.-F.S P{\'e}termann and J-L. R{\'e}my.
\newblock {Arbitrary precision error analysis for computing $\zeta(s)$ with the
  Cohen-Olivier algorithm: complete description of the real case and
  preliminary report on the general case}.
\newblock Rapport de recherche RR-5852, INRIA, 2006.

\bibitem{sage}
W.\thinspace{}A. Stein et~al.
\newblock {\em {S}age {M}athematics {S}oftware}.
\newblock The Sage Development Team, 2013.
\newblock \url{http://www.sagemath.org}.

\bibitem{vdH:stablemult}
J.~van~der Hoeven.
\newblock Making fast multiplication of polynomials numerically stable.
\newblock Technical Report 2008-02, Universit\'e Paris-Sud, Orsay, France,
  2008.

\bibitem{vdH:ball}
J.~van~der Hoeven.
\newblock Ball arithmetic.
\newblock Technical report, HAL, 2009.
\newblock \url{http://hal.archives-ouvertes.fr/hal-00432152/fr/}.

\bibitem{Vepstas2008}
L.~Vep{\v{s}}tas.
\newblock An efficient algorithm for accelerating the convergence of
  oscillatory series, useful for computing the polylogarithm and {H}urwitz zeta
  functions.
\newblock {\em Numerical Algorithms}, 47(3):211--252, 2008.

\end{thebibliography}

\end{document}